\documentclass{fundam}

\usepackage{latexsym}
\usepackage{amsfonts}
\usepackage{amsmath}
\usepackage{amsmath,amssymb,proof1,latexsym}
\usepackage[pdf,arrow,curve]{xy}
\usepackage{etoolbox}

\newcommand{\forces}{\!\Vdash\!}
\newcommand{\imp}{\!\rightarrow\!}

\newcommand{\proves}{\vdash}

\newcommand{\bk}{\ensuremath{\mathbf{K}}}

\newtheorem{notat}{\bf Notation.}

\newtheorem{Comment}{\bf Comment}

\begin{document}

\setcounter{page}{45}
\publyear{22}
\papernumber{2118}
\volume{186}
\issue{1-4}

    \finalVersionForARXIV

\title{Towards Syntactic Epistemic Logic}

\author{Sergei Artemov\thanks{Address for correspondence: The Graduate Center, The City University of New York,
                               365 Fifth Avenue, New York City, NY 10016, USA.}
                  \\
  The Graduate Center, The City University of New York\\
  365 Fifth Avenue, New York City, NY 10016, USA\\
  sartemov@gc.cuny.edu
  }

\runninghead{S. Artemov}{Towards Syntactic Epistemic Logic}

\maketitle

\begin{abstract}
Traditionally, Epistemic Logic represents epistemic scenarios using a single model. This, however, covers only complete descriptions that specify truth values of \textit{all} assertions.  Indeed, many---and perhaps most---epistemic descriptions are not complete. Syntactic Epistemic Logic, {\sf SEL}, suggests viewing an epistemic situation as a set of syntactic conditions rather than as a model. This allows us to naturally capture incomplete descriptions; we discuss a case study in which our proposal is successful.  In Epistemic Game Theory, this closes the conceptual and technical gap, identified by R.~Aumann, between the syntactic character of game-descriptions and semantic representations of games.
\end{abstract}

\section{\bf Introduction}\label{intro}

In this paper, we argue for a paradigm shift in the way that logic and epistemic-related applications -- in particular, game theory -- specify epistemic scenarios.\footnote{The preliminary version of this paper was delivered as an invited talk at the 15th LMPS Congress in 2015 \cite{Art15}.}
Given a verbal description of a situation, a typical epistemic user cherrypicks a ``natural model" (Kripke or Aumann's) and then regards it as a formalization of the original description. This approach carries with it two fundamental deficiencies:
\begin{quote}
I.  It covers only complete descriptions, whereas many (intuitively most) epistemic situations are partially described and cannot be adequately specified by a single model.\footnote{Epistemic logicians have been mostly aware of (I) but this did not stop the wide spread culture of identifying an epistemic scenario with a single Kripke model (or Aumann structure in Game Theory).}
\medskip\par
II. The traditional epistemic reading of Kripke/Aumann models requires common knowledge of the model which restricts their generality and utility even further.
\end{quote}

\subsection{Overspecification}

A typical case of (I)  is the overspecification problem. Consider the following description:
\begin{equation}\label{coin}
\mbox{\it A tossed coin lands heads up. Alice sees the coin, Bob does not.}
\end{equation}
Students in an epistemic logic class normally produce a Kripke {\sf S5}-model of this situation as in Figure~1.

\begin{figure}[!h]
\begin{center}
\mbox{
\begin{xy}
(20,4)*{2};
(-15,4)*{1};
(-15,-4)*{h};
(19,-4)*{\neg h};
(3,3)*{R_B};
(-9,8)*{R_{A,B}};
(26,8)*{R_{A,B}};
(20,0)*+{\bullet}="2";
(-15,0)*+{\bullet}="1";
{\ar "1";"2"};
{\ar "2";"1"};
{\ar@(ul,ur) "1";"1"};
{\ar@(ul,ur)"2";"2"};
\end{xy}
}\end{center}\vspace*{-3mm}
\caption{Model $\mathcal{M}_1$.}
\end{figure}
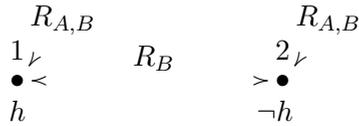
\noindent
In this model, there are two possible worlds 1 and 2, arrows represent indistinguishability relations $R_A$ and $R_B$ between worlds, $h$ is a propositional letter for ``\textit{heads}," and node 1 represents the real world at which $h$ holds.

\medskip
$\mathcal{M}_1$ is a model of (\ref{coin}) which, however, overspecifies (\ref{coin}): in this model there are propositions which are true but do not follow from (\ref{coin}), e.g.,
\begin{itemize}
\itemsep=0.9pt
\item $\bk_A\neg\bk_B h$ - \textit{Alice knows that Bob does not know $h$};\footnote{$\bk_A$ and $\bk_B$ are knowledge modalities for Alice and Bob.}
\item $\bk_B(\bk_A h\vee\bk_A\neg h)$ - \textit{Bob knows that Alice knows whether $h$};
\item \textit{etc.}
\end{itemize}
We will see in Section~\ref{Kripke} that scenario (\ref{coin}) ``as is" does not have a single-model specification at all.

In a situation in which an epistemic scenario is described syntactically but formalized as a model, a completeness analysis relating these two modes is required. For example, the Muddy Children puzzle is given syntactically but then presented as a model tacitly presumed to be commonly known (cf. \cite{FHMV95,MH95,OR94,PvdH07,JvB14}). In Section~\ref{Muddy}, we show that this choice of a specifying model can be justified. However, the Muddy Children case is a fortuitous exception: see Sections~\ref{Muddy} and \ref{games} for more epistemic scenarios without single model specifications.

\medskip
Existing approaches to mitigate overspecification include
\begin{itemize}
\item Supervaluations:  given a syntactically defined situation $\cal S$, assume
\[ \mbox{\it ``$F$ holds in $\cal S$" iff ``$F$ is true in all models of $\cal S$."}
\]
This approach has been dominant in mathematical logic with formal theories as ``situations," and it manifests itself in G\"odel's Completeness Theorem.
\item  Non-standard truth values: Kleene three-valued logic or other, more exotic ways of defining truth values. This approach has generated mathematically attractive models, but it has neither dethroned the supervaluation tradition in mathematical logic, nor changed the ill-founded ``natural model culture" in epistemology.
\end{itemize}
Here we explore the supervaluation approach in epistemology by representing epistemic scenarios in a logical language syntactically and considering the whole class of the corresponding models, not just one cherrypicked model. This also eliminates problem (II).

\section{What is Syntactic Epistemic Logic?}

The name {\em Syntactic Epistemic Logic} was suggested by Robert Aumann (cf. \cite{Aum10}) who identified the conceptual and technical gap between the syntactic character of game descriptions and the predominantly semantic way of analyzing games via relational/partition models.

\medskip
Suppose the initial description ${\cal I}$ of an epistemic situation is syntactic in a natural language. The long-standing tradition in epistemic logic and game theory is then to proceed to a specific epistemic model ${\cal M}_{\cal I}$,
and take the latter as a mathematical definition of ${\cal I}$:
\begin{equation}\label{semantic} \mbox{\it informal description $\cal I$}\ \ \Rightarrow\ \mbox{\it ``natural model'' ${\cal M}_{\cal I}$.}
\end{equation}
Hidden dangers lurk within this process: a syntactic description $\cal I$ may have multiple models and picking one of them (especially declaring it common knowledge) is not generally sound. Furthermore, if we seek an exact specification, then only deductively complete scenarios can be represented (cf. Theorem~\ref{tmain}). Epistemic scenarios outside this group, which include situations with asymmetric and less-than-common knowledge (e.g., mutual knowledge) of conditions, do not have single-model presentations, but can be specified and handled syntactically.

\medskip
Through the framework of {\em Syntactic Epistemic Logic}, {\sf SEL}, we suggest making the syntactic formalization ${\cal S}_{\cal I}$ a formal definition of the situation described by $\cal I$:
\begin{equation}\label{syntactic} \mbox{\it description $\cal I$}\ \Rightarrow\ \mbox{\it syntactic formalization ${\cal S}_{\cal I}$}\Rightarrow \mbox{\it all of ${\cal S}_{\cal I}$'s models.}
\end{equation}
The first step from $\cal I$ to ${\cal S}_{\cal I}$ is formalization and it has its own subtleties which we will not analyze here.

The {\sf SEL} approach (\ref{syntactic}), we argue, encompasses a broader class of epistemic scenarios than a semantic approach (\ref{semantic}).
In this paper, we provide motivations and sketch basic ideas of Syntactic Epistemic Logic. Specific suggestions of general purpose formal systems is a work in progress, cf.~\cite{Art18}.

{\sf SEL} provides a more balanced view of the epistemic universe as being comprised of two inseparable entities, syntactic and semantic. Such a dual view of objects is well-established in mathematical logic where the syntactic notion of a formal theory is supplemented by the notion of a class of all its models. One could expect equally productive interactions between syntax and semantics in epistemology as well.

The definition of a game with epistemic conditions, cf. \cite{Aum76,Aum95}, was originally semantic in a single-model format. In more recent papers (cf. \cite{AA12,Aum10}), Aumann acknowledges the deficiencies of purely semantic formalizations and asks for some kind of ``syntactic epistemic logic" to bridge a gap between the syntactic character of game descriptions and the semantic way of analyzing games.

In this paper, we look at extensive games; the syntactic epistemic approach to strategic games has been tried in \cite{Art14}. However, neither of these papers considers Epistemic Game Theory in its entirety, including probabilistic belief models, cf. \cite{Bra14}; we leave this for future studies.


\section{Logical postulates and derivations}\label{postulates}
We consider the language of classical propositional logic augmented by modalities $\bk_i$, for agent $i$'s knowledge, $i=1,2,\ldots,n$.  For the purposes of this paper, we consider the usual ``knowledge postulates'' (cf. \cite{BRV01,CZ97,FHMV95,MH95,JvB14}) corresponding to the multi-agent modal logic ${\sf S5}_n$:\footnote{The same approach works for other epistemic modal logics.}
\begin{itemize}
\itemsep=0.85pt
\item classical logic postulates and rule {\em Modus Ponens} $A,A\imp B \proves B$;
\item distributivity: $\bk_i(A\imp B)\imp(\bk_i A\imp\bk_i B)$;
\item reflection: $\bk_i A\imp A$;
\item positive introspection: $\bk_i A\imp\bk_i\bk_i A$;
\item negative introspection: $\neg\bk_i A\imp\bk_i\neg\bk_i A$;
\item necessitation rule: $\proves A\ \ \Rightarrow\ \ \proves \bk_i A$.
\end{itemize}
A derivation in ${\sf S5}_n$ is a derivation from ${\sf S5}_n$-axioms by  ${\sf S5}_n$-rules ({\em Modus Ponens} and necessitation). The notation
\begin{equation}\label{ins5}
\proves A
\end{equation}
is used to represent the fact that $A$ is derivable in ${\sf S5}_n$.


\subsection{Derivations from hypotheses}

For a given set of formulas $\Gamma$ (here called ``hypotheses'' or ``assumptions") we consider {\em derivations from $\Gamma$}: assume all ${\sf S5}_n$-theorems, $\Gamma$, and use classical reasoning (rule {\em Modus Ponens}).  The notation
\begin{equation}\label{fromgamma}
\Gamma\proves A
\end{equation}
represents {\em $A$ is derivable from $\Gamma$}.

\medskip
It is important to distinguish the role of necessitation in reasoning without assumptions (\ref{ins5}) and in reasoning from a nonempty set of assumptions (\ref{fromgamma}). In (\ref{ins5}), necessitation can be used freely: what is derived from logical postulates ($\proves A$) is known ($\proves \bk_i A$). In (\ref{fromgamma}), the rule of necessitation is not postulated: if $A$ follows from a set of assumptions $\Gamma$, we cannot conclude that $A$ is known, since $\Gamma$ itself can be unknown. However, for some ``good'' sets of assumptions $\Gamma$, necessitation is a valid rule (cf. $\Gamma_3$ from Example~\ref{E2}, ${\sf MC}_n$ from Section~\ref{Muddy}).

\begin{example}\label{e0}
If we want to describe a situation in which proposition $m$ is known to agent 1, we consider the set of assumptions $\Gamma$:
$$\Gamma = \{\bk_1 m\}.$$
From this $\Gamma$, by reflection principle $ \bk_1 m \imp m$ from ${\sf S5}_n$,
 we can derive $m$,
 \[ \Gamma\proves m .\]
Likewise, we can conclude `1 knows that 1 knows $m$' by using positive introspection:
\[ \Gamma\proves\bk_1\bk_1 m .\]
 However, we cannot conclude that agent 2 knows $m$:
\[ \Gamma\not\proves\bk_2 m .\]
This is rather clear intuitively since when assuming `1 knows $m$,' we do not settle the question of whether `2 knows $m$.'\footnote{A rigorous proof of this non-derivability can be made by providing a counter-model.}
Therefore, there is no necessitation in this $\Gamma$, since we have $\Gamma\proves m$ but $\Gamma\not\proves \bk_2 m$.
\end{example}

\subsection{Common knowledge and necessitation}

We will also use abbreviations: for ``everybody's knowledge''
\[ {\bf E} X = {\bk_1}X\wedge\ldots\wedge {\bk_n}X, \]
and ``common knowledge''
\[ {\bf C}X = \{X,\ {\bf E}X,\ {\bf E}^2 X,\  {\bf E}^3 X,\ \ldots  \}.  \]
As one can see, ${\bf C} X$ is an infinite set of formulas. Since modalities $\bk_i$ commute with the conjunction, ${\bf C}X$ is provably equivalent to the set of all formulas which are $X$ prefixed by iterated knowledge modalities:
\[ {\bf C}X= \{P_1P_2\ldots P_k X \mid k=0,1,2,\ldots,\ \  P_i \in \{ \bk_1,\ldots,\bk_n\}\}.\]
Naturally, \[{\bf C}\Gamma=\bigcup \{{\bf C}F\mid F\in\Gamma\}\] that states ``$\Gamma$ is common knowledge.''

\begin{quote} The set of formulas ${\bf C}X$ emulates common knowledge of $X$ using the conventional modalities $\{ \bk_1,\ldots,\bk_n\}$. This allows us to speak, to the extent we need here, about common knowledge without introducing a special modality and new principles.
\end{quote}

 The following proposition states that the rule of necessitation corresponds to common knowledge of all assumptions. If $\Gamma,\Delta$ are sets of formulas, then $\Gamma\proves\Delta$ means $\Gamma\proves X$ for each $X\in\Delta$.

\begin{proposition}\label{ck=nec}  {\em A set of formulas $\Gamma$ is closed under necessitation if and only if
$\Gamma\proves {\bf C}\Gamma$, i.e., that $\Gamma$ proves its own common knowledge.}
\end{proposition}
\begin{proof} Direction `if.' Assume $\Gamma\proves {\bf C}\Gamma$ and prove by induction on derivations that $\Gamma\proves X$ yields $\Gamma\proves \bk_i X$. For $X$ being a theorem of ${\sf S5}_n$, this follows from the rule of necessitation in ${\sf S5}_n$. For $X\in\Gamma$, it follows from the assumption that $\Gamma\proves {\bf C}X$, hence $\Gamma\proves \bk_i X$. If $X$ is obtained from {\em Modus Ponens}, $\Gamma\proves Y\imp X$ and $\Gamma\proves Y$. By the induction hypothesis, $\Gamma\proves \bk_i(Y\imp X)$ and $\Gamma\proves \bk_i Y$. By the distributivity principle of ${\sf S5}_n$, $\Gamma\proves \bk_i X$.

\medskip
For `only if,' suppose that $\Gamma$ is closed under necessitation and $X\in\Gamma$, hence $\Gamma\proves X$. Using appropriate instances of the necessitation rule in $\Gamma$ we can derive $P_1P_2P_3,\ldots,P_k X$ for each prefix $P_1P_2P_3,\ldots,P_k$ with $P_i$ is one of $\bk_1,\bk_2,\ldots,\bk_n$. Therefore, $\Gamma\proves {\bf C}X$ and $\Gamma\proves {\bf C}\Gamma$.
\end{proof}

\section{Kripke structures and models}\label{Kripke}

A Kripke structure is a convenient vehicle for specifying epistemic assertions via truth values of atomic propositions and the combinatorial structure of the set of global states of the system.
A {\bf Kripke structure} $${\cal M}=\langle W,R_1,R_2,\ldots, \forces\rangle$$ consists of a non-empty set $W$ of possible worlds, ``indistinguishability'' equivalence relations $R_1,R_2,\ldots$ for each agent, and truth assignment  `$\ \forces\ $' of atoms at each world.  The predicate `{\em $F$ holds at $u$}' ($u\forces F$) respects Booleans and reads epistemic assertions as
\[ \mbox{\em $u\forces \bk_i F \ \ \ $ {\em iff} $\ \ \ $ for each state $v\in W$ with $uR_i v$, $v\forces F$ holds}. \]
Conceptually, `{\em agent $i$ at state $u$ knows $F$}' $(u\forces \bk_i F)$ encodes the situation in which $F$ holds at each state indistinguishable from $u$ for agent $i$.

\medskip
A {\bf model} of a set of formulas $\Gamma$ is a pair $({\cal M},u)$ of a Kripke structure ${\cal M}$ and a state $u$ such that all formulas from $\Gamma$ hold at $u$:
\[ {\cal M},u\forces F\ \ \mbox{\em for all $F\in\Gamma$.}\footnote{Here we mean {\em local models} when $\Gamma$ is satisfied at one world. This should not be confused with the more restrictive notion of {\em global models} when $\Gamma$ is satisfied at each world of the model (cf. \cite{BRV01,BvB07,Fit07}). For epistemic purposes, global models do not suffice: for example, a consistent and meaningful situation $\Gamma=\{m,\neg \bk m\}$, i.e., `$m$ holds but is not known' does not have global models.}\]
A pair $({\cal M},u)$ is an {\bf exact model} of $\Gamma$ if
\[ \Gamma\proves F\ \ \Leftrightarrow\ \ {\cal M},u\forces F .\]
An epistemic scenario (a set of ${\sf S5}_n$-formulas) $\Gamma$ admits a {\bf semantic definition} iff $\Gamma$ has an exact model.

\medskip
There is a simple criterion to determine whether $\Gamma$ admits semantic definitions (Theorem~\ref{tmain}) and we argue that ``most" epistemic scenarios lack semantic definitions. These observations provide a justification for Syntactic Epistemic Logic with its syntactic approach to epistemic scenarios.

A formula $F$ follows semantically from $\Gamma$, $$\Gamma\models F,$$ if $F$ holds in each model $({\cal M},u)$ of $\Gamma$.  A well-known fact connecting syntactic derivability from $\Gamma$ and semantic consequence is given by the {\bf Completeness Theorem}\footnote{There are many sources in which the proof of this theorem can be found, e.g., \cite{BRV01,BvB07,CZ97,FHMV95,Fit07,MH95,JvB14}.}:
\[ \Gamma\proves F\ \ \Leftrightarrow\ \ \Gamma\models F.
\]
This fact has been used to claim the equivalence of the syntactic and semantic approaches and to define epistemic scenarios semantically by a model. However, the semantic part of the Completeness Theorem
\[ \Gamma\models F \]
refers to the validity of $F$ in {\bf all} models of $\Gamma$, not in an arbitrary single model.

\medskip
We challenge the model theoretical doctrine in epistemology and show the limitations of single-model semantic specifications, cf. Theorem~\ref{tmain}.

\subsection{Canonical model}\label{canmod}

The Completeness Theorem claims that if $\Gamma$ does not derive $F$, then there is a model $({\cal M},u)$ of $\Gamma$ in which $F$ is false. Where does this model come from?

\medskip
The standard answer is given by the canonical model construction. In any model $({\cal M},u)$ of $\Gamma$, the set of truths $\cal T$ contains $\Gamma$ and is {\it maximal}, i.e., for each formula $F$,
\[ F\in{\cal T}\ \ \ \mbox{ or }\ \ \ \neg F\in{\cal T}. \]
This observation suggests the notion of a {\em possible world} as a maximal set of formulas $\Gamma$ which is consistent, i.e.,
$\Gamma\not\proves\bot$.

\medskip
A {\em canonical model} ${\cal M}({\sf S5}_n)$ of ${\sf S5}_n$ (cf. \cite{BRV01,BvB07,CZ97,FHMV95,Fit07,MH95}) consists of all possible worlds over ${\sf S5}_n$. Accessibility relations are defined on the basis of what is known at each world: for maximal consistent sets $\alpha$ and $\beta$,
\[ \mbox{$\alpha R_i \beta\ \ \ $ iff $\ \ \ \alpha_{\bk_i} \subseteq \beta$} \]
where $$\alpha_{\bk_i}=\{F\mid \bk_iF\in\alpha\},$$
i.e.,
\[ \mbox{\em all facts that are known at $\alpha$ are true at $\beta$}.\]
Evaluations of atomic propositions are defined accordingly:
\[ \alpha\forces p_i\ \ \mbox{ iff }\ \ p_i\in\alpha.\]
The standard Truth Lemma shows that Kripkean truth values in the canonical model agree with possible worlds: for each formula $F$,
\[ \alpha\forces F\ \ \mbox{ iff }\ \ F\in\alpha. \]

The canonical model ${\cal M}({\sf S5}_n)$ of ${\sf S5}_n$ serves as a parametrized universal model for each consistent epistemic scenario $\Gamma$. Given $\Gamma$, by the well-known Lindenbaum construction, extend $\Gamma$ to a maximal consistent set $\alpha$. By definition, $\alpha$ is a possible world in ${\cal M}({\sf S5}_n)$. By the Truth Lemma, all formulas from $\Gamma$ hold in $\alpha$:
\[ {\cal M}({\sf S5}_n),\alpha\ \forces\ \Gamma .\]

\subsection{Deductive completeness}

\begin{definition} A set of ${\sf S5}_n$-formulas $\Gamma$ is {\em deductively complete} if
\[ \Gamma \proves F\ \ \mbox{ or }\ \ \Gamma\proves\neg F .
\]
\end{definition}
\begin{example}\label{E2}
Consider examples in the language of the two-agent epistemic logic ${\sf S5}_2$ with one propositional variable $m$ and knowledge modalities $\bk_1$ and $\bk_2$.
\medskip\par
1. $\Gamma_1 = \{m\}$, where $m$ is a propositional letter. Neither $\bk_1m$ nor $\neg\bk_1m$ is derivable in $\Gamma_1$ and this can be easily shown on corresponding models. Hence $\Gamma_1$ is not deductively complete.\footnote{In classical logic without epistemic modalities, $\Gamma_1$ is deductively complete: for each modal-free formula $F$ of one variable $m$, either $\Gamma_1\proves F$ or $\Gamma_1\proves\neg F$.}
\medskip\par
2. $\Gamma_2 = \{{\bf E}m\}$, i.e., both agents have first-order knowledge of $m$. However, the second-order knowledge assertions, e.g., $\bk_2\bk_1m$, are independent,\footnote{Again, there are easy countermodels.}
\[ \Gamma_2\not\proves \bk_2\bk_1m\ \ \mbox{\ and }\ \ \Gamma_2\not\proves \neg \bk_2\bk_1m .\]
This makes $\Gamma_2$ deductively incomplete.
\medskip\par
3. $\Gamma_3 = {\bf C}m$, i.e., it is common knowledge that $m$. This set is deductively complete. Indeed, first note that, by Proposition~\ref{ck=nec}, $\Gamma_3$ admits necessitation:\footnote{which is not the case for $\Gamma_1$ and $\Gamma_2$.}
\[ \Gamma_3\proves F\ \ \Rightarrow\ \ \Gamma_3\proves \bk_i F,\ \ \ i=1,2.\]
To establish the completeness property: for each formula $F$,
\[ \Gamma_3\proves F\ \ \mbox{ or } \ \ \Gamma_3\proves\neg F, \]
run induction on $F$. The base case when $F$ is $m$ is covered, since $\Gamma_3\proves m$. The Boolean cases are straightforward.
Case $F=\bk_i X$.  If $\Gamma_3\proves X$, then, by necessitation, $\Gamma_3\proves \bk_iX$. If $\Gamma_3\proves\neg X$, then, since {\sf S5} proves $\neg X\imp \neg\bk_i X$, $\Gamma_3\proves\neg \bk_i X$.
\end{example}

\subsection{Semantic definitions and complete scenarios}

The following observation provides a necessary and sufficient condition for semantic definability. Let $\Gamma$ be a consistent set of formulas in the language of ${\sf S5}_n$.\footnote{The same criteria hold for any other normal modal logic which has a canonical model in the usual sense.}

\begin{theorem}\label{tmain} {\em
$\Gamma$ is semantically definable if and only if it is deductively complete.}
\end{theorem}
\begin{proof}
The `only if' direction. Suppose $\Gamma$ has an exact model $({\cal M},u)$, i.e.,
\[ \Gamma\proves F\ \ \ \Leftrightarrow\ \ \ {\cal M},u\forces F. \]
The set of true formulas in $({\cal M},u)$ is maximal: for each formula $F$,
\[ {\cal M},u\forces F\ \ \mbox{ or }\ \ {\cal M},u\forces \neg F, \]
hence $\Gamma$ is deductively complete: for each $F$,
\[ \Gamma\proves F\ \ \mbox{ or } \ \ \Gamma\proves\neg F. \]

The `if' direction. Suppose $\Gamma$ is consistent and deductively complete. Then the deductive closure $\widetilde{\Gamma}$ of $\Gamma$
\[ \widetilde{\Gamma} = \{ F\mid \Gamma\proves F\},\]
is a maximal consistent set, hence an element of the canonical model ${\cal M}({\sf S5}_n)$. We claim that
$({\cal M}({\sf S5}_n),\widetilde{\Gamma})$ is an exact model of $\Gamma$, i.e., for each $F$,
\[ \Gamma\proves F\ \ \ \Leftrightarrow\ \ \ {\cal M}({\sf S5}_n),\widetilde{\Gamma}\forces F. \]
Indeed, if $\Gamma\proves F$, then $F\in\widetilde{\Gamma}$ by the definition of $\widetilde{\Gamma}$. By the Truth Lemma in
${\cal M}({\sf S5}_n)$, $F$ holds at the world $\widetilde{\Gamma}$. If $\Gamma\not\proves F$, then, by deductive completeness of $\Gamma$, $\Gamma\proves \neg F$, hence, as before, $\neg F$ holds at $\widetilde{\Gamma}$, i.e., ${\cal M}({\sf S5}_n),\widetilde{\Gamma}\not\forces F$.
\end{proof}

Theorem~\ref{tmain} shows serious limitations of semantic definitions. Since, intuitively, deductively complete scenarios $\Gamma$ are exceptions, ``most" epistemic situations cannot be defined semantically.\smallskip

In Section~\ref{MuddyExp}, we provide a yet another example of an incomplete but meaningful epistemic scenario, a natural variant of the Muddy Children puzzle, which, by Theorem~\ref{tmain} does not have a semantic definition, but can nevertheless be easily specified and analyzed syntactically.\smallskip

In Section~\ref{games}, we consider an example of an extensive game with incomplete epistemic description which cannot be defined semantically, but admits an easy syntactic analysis.

\section{The Muddy Children puzzle}\label{Muddy}
Consider the standard Muddy Children puzzle, which is formulated syntactically.
\begin{quote} {\em A group of $n$ children meet their father after playing in the mud. Their father notices that $k>0$ of the children have mud  on their foreheads. The children see everybody else's foreheads, but not their own. The father says: ``some of you are muddy,''  then adds: ``Do any of you know that you have mud on your forehead? If you do, raise your hand now.'' No one raises a hand. The father repeats the question, and again no one moves.
After exactly $k$ repetitions, all children with muddy foreheads raise their hands simultaneously. Why?}
\end{quote}

\subsection{Standard syntactic formalization}\label{SynForm}
This can be described in ${\sf S5}_n$ with modalities $\bk_1, \bk_2,\ldots,\bk_n$ for the children's knowledge and atomic propositions $m_1,m_2,\ldots, m_n$ with $m_i$ stating ``child $i$ is muddy.'' The {\bf initial configuration}, which we call ${\sf MC}_n$, includes common knowledge assertions of the following assumptions:
\medskip\par
1. \textit{Knowing about the others}:
\[ \bigwedge_{i\neq j}[\bk_i(m_j)\vee \bk_i(\neg m_j)].
\]

2. \textit{Not knowing about themselves}:
\[ \bigwedge_{i=1,\ldots,n}[\neg\bk_i(m_i)\wedge \neg\bk_i(\neg m_i)].
\]
Transition from the verbal description of the situation to ${\sf MC}_n$ is a straightforward formalization of a given syntactic description to another, logic friendly syntactic form.

\subsection{Semantic solution}\label{SemSol}
In the  standard semantic solution, the set of assumptions ${\sf MC}_n$ is replaced by a Kripke model: $n$-dimensional cube $Q_n$ (\cite{FHMV95,MH95,OR94,PvdH07,JvB14}). To keep things simple, we consider the case $n=k=2$.

 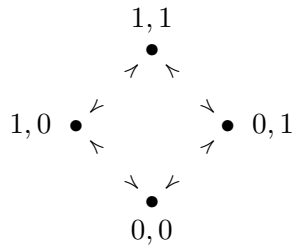
\begin{figure}[!h]
 \vspace{2mm}
\begin{center}
\mbox{
\begin{xy}
(-6,0)*{1,0};
(10,14)*{1,1};
(10, -14)*{0,0};
(26,0)*{0,1};
(0,0)*+{\bullet}="D"; 
(10,10)*+{\bullet}="E"; 
(10,-10)*+{\bullet}="F"; 
(20,0)*+{\bullet}="G"; 
{\ar@{.>} "D";"E"};
{\ar@{.>} "E";"D"};
{\ar "D";"F"};
{\ar "F";"D"};
{\ar@{.>} "F";"G"};
{\ar@{.>} "G";"F"};
{\ar"E";"G"};
{\ar"G";"E"};
\end{xy}
}\vspace*{-3mm}
\end{center}
\caption{Model $Q_2$.}
 \end{figure}

 \eject

 Logical possibilities for the truth value combinations\footnote{$1$ standing for `true' and $0$ for `false'} of $(m_1, m_2)$, namely (0,0), (0,1), (1,0), and (1,1) are declared possible worlds. There are two indistinguishability relations denoted by solid arrows (for 1) and dotted arrows (for 2).
It is easy to check that conditions 1 (knowing about the others) and 2 (not knowing about themselves) hold at each node of this model. Furthermore, $Q_2$ is assumed to be commonly known.

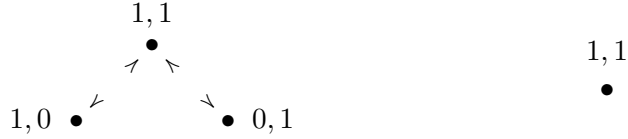
\begin{figure}[!h]
\begin{center}
\mbox{
\begin{xy}
(-6,0)*{1,0};
(10,14)*{1,1};
(26,0)*{0,1};
(0,0)*+{\bullet}="D"; 
(10,10)*+{\bullet}="E"; 
(20,0)*+{\bullet}="G"; 
{\ar@{.>} "D";"E"};
{\ar@{.>} "E";"D"};
{\ar "E";"G"};
{\ar "G";"E"};
(70,9)*{1,1};
(70,4)*+{\bullet}="F";
\end{xy}
}
\end{center}\vspace*{-3mm}
\caption{Models $\mathcal{M}_2$ and $\mathcal{M}_3$.}
 \end{figure}

After the father publicly announces $m_1\vee m_2$, node $(0,0)$ is no longer possible and model $\mathcal{M}_2$ now becomes common knowledge.
Both children realize that in $(1,0)$, child 2 would know whether (s)he is muddy (no other 2-indistinguishable worlds), and in $(0,1)$, child 1 would know whether (s)he is muddy.
After both children answer ``no" to whether they know what is on their foreheads, worlds $(1,0)$ and $(0,1)$ are no longer possible, and each child eliminates them. The only remaining logical possibility here is model ${\cal M}_3$.
Now both children know that their foreheads are muddy.

\subsection{Justifying the model}\label{Just}

The semantic solution in Section~\ref{SemSol} adopts $Q_n$ as a semantic equivalent of a theory ${\sf MC}_n$. Can this choice of the model be justified? In the case of Muddy Children, the answer is `yes.'

\medskip
Let $u$ be a node at $Q_n$. i.e., $u$ is an $n$-tuple of $0$'s and $1$'s and $u\forces m_i$ iff $i$'s projection of $u$ is $1$.
Naturally,  $u$ is represented by a formula $\pi(u)$:
\[ \pi(u) = \bigwedge \{m_i\mid u\forces m_i\} \wedge \bigwedge \{\neg m_i\mid u\forces \neg m_i\} .\]
It is obvious that $v\forces \pi(u)$ iff $v=u$.

\medskip
By ${\sf MC}_n(u)$ we understand the Muddy Children scenario with specific distribution of truth values of $m_i$'s corresponding to $u$:
\[ {\sf MC}_n(u) = {\sf MC}_n \cup \{ \pi(u) \}.\]
So, each specific instance of Muddy Children is formalized by an appropriate ${\sf MC}_n(u)$.

\begin{theorem}\label{tmc} {\em Each instance ${\sf MC}_n(u)$ of Muddy Children is deductively complete and $(Q_n,u)$ is its exact model
\[ {\sf MC}_n(u) \proves F\ \ \ \mbox{ iff }\ \ \ \ Q_n,u\forces F .\]
}
\end{theorem}

\noindent \textbf{Proof:}\footnote{We have chosen to present a syntactic proof of Theorem~\ref{tmc}.  A semantic proof that makes use of bi-simulations can also be given.}
The direction `only if' claims that $(Q_n,u)$ is a model for ${\sf MC}_n(u)$ is straightforward. First, $Q_n$ is an ${\sf S5}_n$-model and all principles of ${\sf S5}_n$ hold everywhere in $Q_n$. It is easy to see that principles `knowing about the others' and `not knowing about himself' hold at each node. Furthermore, as $\pi(u)$ holds at $u$, everything that can be derived from ${\sf MC}_n(u)$ holds at $u$.

\medskip
To establish the `if' direction, we first note that, by Proposition~\ref{ck=nec}, necessitation is admissible in ${\sf MC}_n$: for each $F$,
\[ {\sf MC}_n \proves F\ \ \Rightarrow\ \ \ {\sf MC}_n \proves \bk_i F. \]
The theorem now follows from the statement ${\cal S}(F)$:
\begin{quote} {\em for all nodes $u\in Q_n$,
\[ Q_n,u\forces F \ \ \ \ \Rightarrow\ \ \ {\sf MC}_n\proves \pi(u)\imp F \]
and
\[ Q_n,u\forces \neg F \ \ \ \Rightarrow\ \ {\sf MC}_n\proves \pi(u)\imp \neg F . \]
}
\end{quote}
We prove that ${\cal S}(F)$ holds for all $F$ by induction on $F$.

\medskip
The case $F$ is one of the atomic propositions $m_1,m_2,\ldots,m_n$ is trivial since ${\sf MC}_n\proves \pi(u)\imp m_i$, if $u\forces m_i$ and ${\sf MC}_n\proves \pi(u)\imp \neg m_i$, if $u\forces \neg m_i$. The Boolean cases are also straightforward.

The case $F=\bk_i X$. Consider the node $u^i$ which differs from $u$ only at the $i$-coordinate. Without a loss of generality, we may assume that $u\forces m_i$ and $u^i\forces\neg m_i$; the alternative $u\forces \neg m_i$ and $u^i\forces m_i$ is similar.

\medskip
Suppose $Q_n,u\forces \bk_i X$. Then $Q_n,u\forces X$ and $Q_n,u^i\forces X$. By the induction hypothesis,
\[ \mbox{${\sf MC}_n\proves \pi(u)\imp X\ \ $ and $\ \ {\sf MC}_n\proves \pi(u^i)\imp X$. }\]
By the rules of logic (splitting premises)
\[ \mbox{${\sf MC}_n\proves \pi(u)_{-i}\imp (m_i \imp X)\ \ $ and $\ \ {\sf MC}_n\proves \pi(u)_{-i}\imp (\neg m_i\imp  X)$, }\]
where $\pi(v)_{-i}$ is $\pi(v)$ without its $i$-th coordinate\footnote{Formally, $\pi(v)_{-i}=\bigwedge \{m_j\mid v\forces m_j, \ j\neq i\}\wedge \bigwedge \{\neg m_j\mid v\forces \neg m_j, \ j\neq i\}$. }. By further reasoning,
\[ \mbox{${\sf MC}_n\proves \pi(u)_{-i}\imp  X$. }\]
By necessitation in ${\sf MC}_n$, and distributivity,
\[ \mbox{${\sf MC}_n\proves \bk_i\pi(u)_{-i}\imp  \bk_i X$. }\]
By `knowing about the others' principle, and since $\pi(u)_{-i}$ contains only atoms other them $m_i$,
\[ \mbox{${\sf MC}_n\proves \pi(u)_{-i}\imp  \bk_i\pi(u)_{-i}$, }\]
hence
\[ \mbox{${\sf MC}_n\proves \pi(u)_{-i}\imp  \bk_i X$, }\]
and
\[ \mbox{${\sf MC}_n\proves \pi(u)\imp  \bk_i X$. }\]

Now suppose $Q_n,u\forces \neg \bk_i X$. Then $Q_n,u\forces \neg X$ or $Q_n,u^i\forces \neg X$. By the induction hypothesis,
\[ \mbox{${\sf MC}_n\proves \pi(u)\imp \neg X\ \ $ or $\ \ {\sf MC}_n\proves \pi(u^i)\imp \neg X$. }\]
In the former case we immediately get ${\sf MC}_n\proves \pi(u)\imp \neg \bk_i X$, by reflection $\neg X\imp \neg \bk_i X$. So, consider the latter, i.e.,
${\sf MC}_n\proves \pi(u^i)\imp \neg X$.
As before,
\[ {\sf MC}_n\proves \pi(u)_{-i}\imp (\neg m_i\imp  \neg X). \]
By contrapositive,
\[ {\sf MC}_n\proves \pi(u)_{-i}\imp (X \imp  m_i). \]
By necessitation and distribution,
\[ {\sf MC}_n\proves \bk_i \pi(u)_{-i}\imp (\bk_i X \imp  \bk_i m_i). \]
By `knowing about others,' as before,
\[ {\sf MC}_n\proves \pi(u)_{-i}\imp (\bk_i X \imp  \bk_i m_i). \]
By `not knowing about himself,' ${\sf MC}_n\proves \neg \bk_i m_i$, hence
\[ {\sf MC}_n\proves \pi(u)_{-i}\imp \neg \bk_i X, \]
and
\[ {\sf MC}_n\proves \pi(u)\imp \neg \bk_i X. \smallskip \]

\vspace*{-2mm}

As we see, in the case of Muddy Children given by a syntactic description, ${\sf MC}_n(u)$, picking one ``natural model'' $(Q_n,u)$ could be justified. However, in a general setting, the approach
 \[  \mbox{\em given a syntactic description, pick a ``natural model''} \]
is intrinsically flawed: by Theorem~\ref{tmain}, in many (intuitively, most) cases, there is no model description at all. Furthermore, if there is a ``natural model,'' a completeness analysis in the style of what we did for ${\sf MC}_n$ in Theorem~\ref{tmc} is required.

\eject
\subsection{Incomplete scenario:  Muddy Children Explicit}\label{MuddyExp}

Here is a natural modification, ${\sf MCE}_{n,k}$, of the standard Muddy Children.
\begin{quote} \textit{A group of $n$ children meet their father after playing in the mud. Each child sees everybody else's foreheads. The father says: ``$k$ of you are muddy"
after which it becomes common knowledge that all children know whether they are muddy. Why?}
\end{quote}
This description does not specify whether children initially know if they are muddy; hence the initial description of ${\sf MCE}_{n,k}$ is, generally speaking, not complete\footnote{In particular, prior to father's announcement ${\sf MCE}_{2,2}$ does not specify whether $\bk_1 m_1$ holds or not.}. By Theorem~\ref{tmain}, the initial ${\sf MCE}_{2,2}$ is not semantically definable. Therefore,  ${\sf MCE}_{2,2}$ cannot be treated by ``natural model" methods.

\medskip
However, here is a syntactic analysis of ${\sf MCE}_{n,k}$ which can be shaped as a formal logical reasoning within an appropriate extension of ${\sf S5}_n$.
\begin{quote} After father's announcement, each child knows that if she sees $k$ muddy foreheads, then she is not muddy, and if she sees $k\!-\!1$ muddy foreheads, she is muddy: this secures that each child knows whether she is muddy. Moreover, everybody can reflect on this reasoning and this makes it common knowledge that each child knows whether she is muddy.
\end{quote}

\subsection{Some additional observations}

If we want to go beyond complete epistemic scenarios, we need a mathematical apparatus to handle classes of models, and not just single models. The format of syntactic specifications in some version of the modal epistemic language is a viable candidate for such an apparatus.

\medskip
The traditional model solution of ${\sf MC}_n$ without completeness analysis uses a strong additional assumption -- common knowledge of a specific model $Q_n$ and hence, strictly speaking, does not resolve the original Muddy Children puzzle; it rather corresponds to a different scenario of a more tightly controlled epistemic states of agents, e.g.,
\begin{quote} {\em A group of robots programmed to reason about model $Q_n$ meet their programmer after playing in the mud. ...}
\end{quote}

One could argue that the given model solution of ${\sf MC}_n$ actually codifies some deductive solution in the same way that geometric reasoning is merely a visualization of a rigorous derivation in some sort of axiom system for geometry. This is a valid point which can be made scientific within the framework of Syntactic Epistemic Logic.

\section{Syntactic Epistemic Logic and games}\label{games}
Consider a variant \textit{Centipede Lite}, {\sf CL}, of the well-known Centipede game (cf. \cite{OR94}) with risk-averse rational players Alice and Bob. No cross-knowledge of rationality, let alone common knowledge, is assumed!

\begin{figure}[!h]
\begin{center}\mbox{
\begin{xy}
(-80,10)="A";
(-80,6)*{2,1};
(-80,20)*+{\bullet}="B";
{\ar "B";"A"};
(-77,23)*{1(A)};
(-65,10)="C";
(-65,6)*{1,4};
(-65,20)*+{\bullet}="D";
{\ar "D";"C"};
(-62,23)*{2(B)};
{\ar "B";"D"};
(-50,10)="E";
(-50,6)*{4,3};
(-50,20)*+{\bullet}="F";
{\ar "F";"E"};
(-47,23)*{3(A)};
{\ar "D";"F"};
(-35,10)="G";
(-32,20)*{3,6};
(-35,20)*+{}="H";
{\ar "F";"H"};
\end{xy}
}\end{center}\vspace*{-4mm}
\caption{Centipede game tree}
\end{figure}
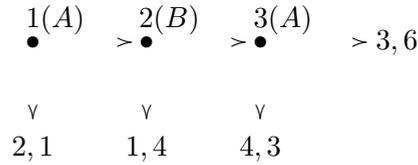

\medskip\noindent
{\sf CL} admits the following rigorous analysis.
 \begin{quote} At 3, Alice plays down. At 2, Bob plays down because he is risk-averse and cannot rule out that Alice plays down at 3 (since it is true). At 1, Alice plays down because she cannot rule out Bob's playing down at 2. So, {\sf CL} has the so-called Backward Induction solution {\it ``down at each node."}
\end{quote}
{\sf CL} is not complete (epistemic assumptions, such as
{\it Bob knows that Alice plays ``across" at 3}, are not specified), hence
{\sf CL} cannot be defined by a single Kripke/Aumann model.

\section{Incomplete and complete scenarios}

How typical are deductively incomplete epistemic scenarios? We argue that this is the rule rather than the exception. Epistemic conditions more flexible than common knowledge of the game and rationality (mutual knowledge of rationality, asymmetric epistemic assumptions when one player knows more than the other, etc.) lead to semantic undefinability.

Semantically non-definable scenarios are the ``dark matter'' of the epistemic universe: they are everywhere, but cannot be visualized as a single model. The semantic approach does not recognize these ``dark matter'' scenarios; {\sf SEL} deals with them syntactically.

The question remains: how manageable are semantic definitions of deductively complete scenarios?

\subsection{Cardinality and knowability issue}
Models of complete $\Gamma$'s provided by Theorem~\ref{tmain} are instances of the canonical model ${\cal M}({\sf S5}_n)$ at nodes $\widetilde{\Gamma}$ corresponding to $\Gamma$. This generic solution is, however, not satisfactory because of the highly nonconstructive nature of the canonical model ${\cal M}({\sf S5}_n)$.

As was shown in \cite{Aum99}, the canonical model ${\cal M}({\sf S5}_n)$ for any $n\geq 1$ has continuum-many possible worlds even with just one propositional letter. This alone renders models $({\cal M}({\sf S5}_n),\widetilde{\Gamma})$ not knowable under any reasonable meaning of ``known.'' The canonical model for ${\sf S5}_n$ is just too large to be considered known and hence does not {\it a priori} satisfy the knowability of the model requirement II from Section~\ref{intro}.

This observation suggests that the question about existence of an epistemically acceptable (``known'') model for a given deductively complete set $\Gamma$ requires a case-by-case consideration.

\subsection{Complexity considerations}

Epistemic models of even simple and complete scenarios can be prohibitively large compared to their syntactic descriptions. For example, the Muddy Children model $Q_n$ is exponential in $n$ whereas its syntactic description ${\sf MC}_n$ is quadratic in $n$.

Consider a real-life epistemic situation after the cards have been initially dealt in the game of poker. One can show that for each distribution of cards, its natural syntactic description in epistemic logic is deductively complete (\cite{AN15}) and hence admits a model characterization. Moreover, it has a natural finite model of the type given in \cite{FHMV95} with hands as possible worlds and with straightforward knowledge relations. However, with 52 cards and 4 players there are over $10^{24}$ different combinations of hands. This yields that explicit formalization of the model not practical. Players reason using concise syntactic descriptions of the rules of poker and of its ``large" model in the natural language, which can also be syntactically formalized in some kind of extension of epistemic logic.

In this and some other real life situations, models are prohibitively large whereas appropriate syntactic descriptions can be quite manageable.

\section{Further observations}
An interesting question is why the traditional semantic approach, despite its aforementioned shortcomings, produces correct answers in many situations. One of possible reasons for this is \textit{pragmatic self-limitation}.

Given a syntactic description $\cal D$, we intuitively seek a solution that logically follows from $\cal D$. Even if we reason on  a ``natural model'' of $\cal D$, normally overspecified, we try not to use features of the model that are not supported by $\cal D$. If we conclude a property $P$ by such self-restricted reasoning about the model, then $P$ indeed logically follows from $\cal D$.
\begin{quote}  This situation resembles Geometry, in which we reason about ``models", i.e., combinations of triangles, circles, etc., but have a rigorous system of postulates in the background. We are trained not to venture beyond given postulates even in informal reasoning.
\end{quote}
Such an {\em ad hoc} pragmatic approach needs a scientific foundation, which could be provided within the framework of Syntactic Epistemic Logic.

\section{Syntactic Epistemic Logic suggestions}

The Syntactic Epistemic Logic suggestion, in brief, is {\em to make an appropriate syntactic formalization of an epistemic scenario its formal specification}. This extends the scope of scientific epistemology and offers a remedy for two principal weaknesses of the traditional semantic approach. The reader will recall that those weaknesses were the restricting single model requirement and a hidden assumption of the common knowledge of this model.

 {\sf SEL} suggests a way to handle incomplete scenarios which have rigorous syntactic descriptions (cf.
Muddy Children Explicit, Centipede Lite, etc.).

{\sf SEL} offers a scientific framework for resolving the tension, identified by R.~Aumann \cite{Aum10}, between a syntactic description and its hand-picked model. If, given a syntactic description $\Gamma$ we prefer to reason on a model $\cal M$, we have to establish completeness of $\Gamma$ with respect to $\cal M$.

Appropriate syntactic specifications could also help to handle situations for which natural models exist but are too large for explicit presentations.

{\sf SEL}  can help to extend Epistemic Game Theory to less restrictive epistemic conditions. A broad class of epistemic scenarios does not define higher-order epistemic assertions and rather addresses individual knowledge, mutual and limited-depth knowledge, asymmetric knowledge, etc. and hence is deductively incomplete and has no exact single model characterizations. However, if such a scenario allows a syntactic formulation, it can be handled scientifically by a variety of mathematical tools, including reasoning about its models.


Since the basic object in {\sf SEL} is a syntactic description of an epistemic scenario rather than a specific model, there is room for a new syntactic theory of updates and belief revision.

\subsection*{Acknowledgements}
The author is grateful to Adam Brandenburger, Alexandru Baltag, Johan van Benthem, Robert Constable, Melvin Fitting, Vladimir Krupski, Anil Nerode, Elena Nogina, Eoin Moore, Vincent Peluce, Tudor Protopopescu, Bryan Renne, Richard Shore, and Cagil Tasdemir for useful discussions. Special thanks to Karen Kletter for editing early versions of this text.

\end{document}